\documentclass[letterpaper,10pt]{article}

\newcommand{\rurl}[1]{\href{https://#1/}{\nolinkurl{#1}}}

\usepackage[svgnames]{xcolor}
\usepackage[longnamesfirst]{natbib}
\PassOptionsToPackage{backref=page}{hyperref}

\usepackage{xstring}
\usepackage{amsmath}
\usepackage{amssymb}
\usepackage{amsthm}
\usepackage{thm-restate}
\usepackage{mathabx}
\usepackage{mathtools}
\usepackage{todonotes}
\usepackage{enumitem}
\usepackage[ruled,boxed,noend]{algorithm2e}
\usepackage{hyperref}

\usepackage{tikz}

\def\titl{Sparse polynomial interpolation and division in soft-linear
time}
\def\auts{Pascal Giorgi, Bruno Grenet, Armelle Perret du Cray, Daniel S. Roche}

\hypersetup{
  pdfauthor={\auts},
  pdftitle={\titl},
  colorlinks,
  linkcolor=DarkBlue,
  citecolor=DarkGreen,
  urlcolor=DarkBlue,
  bookmarksnumbered,
  pdfpagelabels=true,
}

\usepackage[nameinlink,capitalise]{cleveref}

\SetAlgoVlined
\DontPrintSemicolon
\SetKw{break}{break}
\SetKwInOut{Input}{Input}
\SetKwInOut{Output}{Output}
\SetKwFor{Loop}{repeat}{}{}

\newcommand{\doi}[1]{doi: \href{http://dx.doi.org/#1}{\path{#1}}}

\newtheorem{theorem}{Theorem}[section]
\newtheorem{fact}[theorem]{Fact}

\newtheorem{cor}[theorem]{Corollary}

\crefname{fact}{Fact}{Facts}

\numberwithin{equation}{section}

\newcommand{\ZZ}{\ensuremath{\mathbb{Z}}}

\newcommand{\GF}[1]{\ensuremath{\mathbb{F}_{#1}}}

\newcommand{\bnd}[2]{\ensuremath{#1\mathopen{}\left(#2\right)\mathclose{}}}
\newcommand{\oh}[1]{\bnd{O}{#1}}
\newcommand{\softoh}[1]{\bnd{\widetilde{O}}{#1}}

\newcommand{\poly}{\ensuremath{\mathsf{poly}}}
\newcommand{\polylog}{\ensuremath{\mathsf{polylog}}}

\newcommand{\fail}{\textsc{fail}\xspace}

\renewcommand{\backref}[1]{Referenced on %
  \expandarg\StrCount{#1}{,}[\ncommas]%
  \ifthenelse{\ncommas = 0}{page #1}%
  {pages \StrBefore[\ncommas]{#1}{,}\ and\StrBehind[\ncommas]{#1}{,}}%
.}

\renewcommand{\titl}{Random primes in arithmetic progressions}
\hypersetup{pdftitle={\titl}}
\usepackage{listings}

\usepackage{xcolor}
\definecolor{commentgreen}{RGB}{2,112,10}
\definecolor{eminence}{RGB}{108,48,130}
\definecolor{weborange}{RGB}{255,165,0}
\definecolor{frenchplum}{RGB}{129,20,83}
\usepackage{graphicx}

\title{\titl}
\author{
Pascal Giorgi\\
{\small LIRMM, Univ. Montpellier, CNRS}\\
{\small\rurl{www.lirmm.fr/~giorgi}}
\and
Bruno Grenet\\
{\small LIRMM, Univ. Montpellier, CNRS}\\
{\small\rurl{www.lirmm.fr/~grenet/}}
\and
Armelle Perret du Cray\\
{\small LIRMM, Univ. Montpellier, CNRS} \\
{\small\rurl{www.lirmm.fr/armelle-perret-du-cray}}
\and
Daniel S.\ Roche\\
{\small United States Naval Academy}\\
{\small\rurl{www.usna.edu/cs/roche}}
}

\begin{document}
\maketitle

\begin{abstract}
  We describe a straightforward method to generate a random prime $q$
  such that the multiplicative group $\GF{q}^*$ also has a random large
  prime-order subgroup. The described algorithm also yields this order
  $p$ as well as a $p$'th primitive root of unity $\omega$.
  The methods here are efficient asymptotically, but due to large
  constants may not be very useful in practical settings.
\end{abstract}

\section{Introduction}

In various contexts, for example in sparse polynomial evaluation and
interpolation algorithms, it is necessary to have a finite field
$\GF{q}$ that admits an order-$p$ multiplicative subgroup with generator
$\omega$. There are typically some non-divisibility properties both on
the field size $q$ and the subgroup order $p$.

In this note, we briefly sketch efficient algorithms to
probabilistically generate such
$q,p,\omega$ tuples. The results are neither surprising to practitioners
in this area, nor are they particularly original. However, we have found
them useful, and so decided to publish in this short note with complete
proofs.

\section{Statement of results}

Our approach to produce triples $(p,q,\omega)$ such that $\omega$ generates an order-$p$
multiplicative subgroup of $\GF{q}^\times$ is straightforward. We first sample integers $p$
until a prime number is found. Then, we need to find a prime number $q$ such that $p\divides(q-1)$, 
that is $q$ is in the arithmetic progression $\{ap+1:a\ge 1\}$, and such that $q=\poly(p)$. Again,
we sample integers $a$ until $ap+1$ is prime. Finally, we sample elements $\zeta\in\GF{q}^\times$
until $\zeta^{(q-1)/p}\neq 1$ and return $(p,q,\zeta^{(q-1)/p})$.

Each step is justified by the abondance of \emph{good} integers or elements of $\GF{q}^\times$.
The computation of $p$ relies on an effective version of the prime number theorem of \citet{RS62}. The
computation of $q$ relies on effective versions of Dirichlet's theorem (or more precisely 
Bombieri-Vinogradov theorem) of \citet{AH15} and \citet{Sed18}. The computation of $\omega$
relies on the fact that there are at most $(q-1)/p$ values $\zeta$ such that $\zeta^{(q-1)/p} = 1$.

Our approach is closely related to similar results in \citeauthor{arn16}'s Ph.D.
thesis \citeyearpar{arn16}. We replace a constant probability of success by an
arbitrary high probability of success, and we use better bounds.

\begin{theorem}\label{lem:mcpap}
  There exists an explicit Monte Carlo algorithm which, given a bound
  $\lambda\ge\frac{2^{58}}{\epsilon^2}$, produces a triple
  $(p,q,\omega)$ that has the following properties with probability at least
  $1-\epsilon$, and return \fail otherwise:
  \begin{itemize}
  \item $p$ is uniformly distributed amongst the primes of $(\lambda,2\lambda)$;
  \item $q\le\lambda^6$ is a prime such that $p\divides(q-1)$;
  \item $\omega$ is a $p$-primitive root of unity in $\GF q$;
  \end{itemize}
  Its worst-case bit complexity is $\polylog(\lambda)$.
\end{theorem}

An additional requirement in some situations is that the prime $q$ does not
divide an (unknown!) large integer. This is achieved by taking $\lambda$
sufficiently large.

\begin{theorem}\label{lem:mcpap2}
  Let $K$ be an unknown integer, and let $(p,q,\omega)$ a triple produced by
  the algorithm of \cref{lem:mcpap} on some input $\lambda$.  If
  $\lambda\ge\max(\frac{2^{56}}{\mu^2}, \sqrt[5]{\frac{48}{\mu}\ln K})$, the
  probability that $q$ divides $K$ is at most $\mu$.
\end{theorem}

The large constant in the statement of \cref{lem:mcpap} is required in order
to get rigorous unconditional complexity bounds. Yet, it makes the algorithm
not very practical because of the bit-length of the primes produced. 
\Cref{sec:experiments} presents experimental results indicating the results
actually hold for much smaller values of $\lambda$.

\section{Proofs}

To construct a field $\GF q$ with a $p$-PRU $\omega$, we first need to generate random prime numbers. The well-known technique for this is to sample random integers and test them for primality. In order to get Las Vegas algorithm, we rely on the celebrated AKS algorithm.

\begin{fact}[\citet{AKS04}]\label{fact:primetest}
There is a deterministic algorithm that, given any integer $n$,
determines whether $n$ is prime or composite and has bit complexity
$\polylog(n)$.
\end{fact}

While the original bit complexity was $\softoh{\log^{10.5} n}$, this has
been subsequently improved to $\softoh{\log^6 n}$ in a revised version
by \citet{LP11}. In practice, a better option is to use the Monte Carlo
Miller-Rabin primality test which has a worst-case bit complexity of
$\softoh{\log^2 n}$ but a low probability of incorrectly reporting that
a composite number is prime \citep{Rab80}.

No fast deterministic algorithm is known to \emph{construct} a prime number with a given bit length $b$. However, sampling random $b$-bit integers and testing their primality using AKS algorithm
results in a \emph{Las Vegas} randomized algorithm. The expected running time relies on the fact that there are at least $\Omega(2^b/b)$ primes with $b$ bits. We recall some more precise bounds. 

\begin{fact}[\citet{RS62}]
For $\lambda\ge 21$, there exist at least $\frac{3}{5}\lambda\ln\lambda$ prime numbers between $\lambda$ and $2\lambda$.
\end{fact}

Once we have a prime number $p$, we want to find a prime number $q$ in the
arithmetic progression $p+1$, $2p+1$, $3p+1$, \dots{} 
Dirichlet's theorem says that, \emph{asymptotically}, the distribution of
primes in this arithmetic progression is the
same as the distribution of primes in $\ZZ$. The Bombieri-Vinogradov theorem 
refines it with bounds on the error terms. This indicates that a good
strategy to generate $q$ is simply to pick a random (even) positive integer $k$
and test whether $pk + 1$ is prime, repeating until a prime of that form
is found.

The question is, how large should $k$ be in the strategy above in order to
guarantee a reasonable chance of success? Recent results of \citet{AH15} 
and \citet{Sed18} give explicit bounds for the Bombieri-Vinogradov theorem.

\begin{fact}[{\citet[Corollary~1.5]{Sed18}}]\label{fact:explicitBV}
Let $\pi(x)$ denote the number of prime numbers $\le x$, $\pi(x;m,a)$ the number
of prime numbers $\le x$ that are congruent to $a$ modulo $m$, and $\ell(x)$ the 
smaller prime divisor of $x$. Then for any $\gamma\ge 4$ and $\lambda_1\le\lambda_2\le\gamma^{1/2}$,
\begin{multline*}
\sum_{\substack{m\le\lambda_2\\\ell(m)>\lambda_1}}
    \max_{2\le y\le\gamma} \max_{a:\gcd(a,m)=1}
        \left| \pi(y;m,a)-\frac{\pi(y)}{\phi(m)}\right| \\
  \le 122.77\left(14\frac{\gamma}{\lambda_1}+4\gamma^{1/2}\lambda_2+15\gamma^{2/3}\lambda_2^{1/2}+4\gamma^{5/         6}\ln(\frac{\lambda_2}{\lambda_1})\right)(\ln\gamma)^{7/2}.
\end{multline*}
\end{fact}

From this fact, we obtain the following probabilistic result on the number of primes in an arithmetic
progression. 

\begin{cor}\label{fact:arprog}
Let $0<\epsilon<\frac{1}{2}$ and $\lambda\ge\frac{2^{54}}{\epsilon^2}$, and $p$ be a random prime from $(\lambda, 2\lambda)$. Then with probability at least $1-\epsilon$, the number of prime numbers $q\le\lambda^6$ of the form $q = ap+1$ is $\ge \lambda^5/(24\ln\lambda)$.
\end{cor}

\begin{proof}
We apply \cref{fact:explicitBV} with $\lambda_1 = \lambda$, $\lambda_2 = 2\lambda$ and $\gamma = \lambda^6$. We note that the sum is over the prime numbers (since $\ell(m) > \lambda_1 \ge m/2$). We then simplify it by choosing $y=\gamma$ and $a = 1$ in the formula, which can only make the sum smaller. Then
\[\sum_{\substack{\lambda<p<2\lambda\\p\text{ prime}}}
    \left|\pi(\lambda^6;p,1)-\frac{\pi(\lambda^6)}{p-1}\right|
  \le 1.09\cdot 10^6(\lambda^5+1.27\lambda^{4.5}+0.48\lambda^4)(\ln\lambda)^{7/2}.
\]
For $\lambda\ge 2^8$, the sum is bounded by $1.2\cdot 10^6\lambda^5(\ln\lambda)^{7/2}$.
Now we wish to count the \emph{bad} primes in $(\lambda,2\lambda)$ such that $\pi(\lambda^6;p,1) \le \lambda^5/       (24\ln\lambda)$. Since $\pi(\lambda^6)\ge \lambda^6/(6\ln\lambda)$, if $p$ is a bad prime, then 
$\pi(\lambda^6)/(p-1) \ge \pi(\lambda^6;p,1)$ and since $p-1\le2\lambda$,
\[ \left|\pi(\lambda^6;p,1)-\frac{\pi(\lambda^6)}{p-1}\right| \ge \frac{\lambda^6/(6\ln\lambda)}{p-1}-                \frac{\lambda^5}{24\ln\lambda} \ge \frac{\lambda^5}{24\ln\lambda}.\]
If there are $k$ bad primes, then the sum is at least $k\lambda^5/24\ln\lambda$. Using the previous bound on the sum, we get the bound
\[k\le \frac{1.2\cdot 10^6\lambda^5(\ln\gamma)^{7/2}}{\lambda^5/(24\ln\lambda)} 
    = 2.88\cdot 10^7(\ln\lambda)^{9/2}.\]
Since there are at least $\frac{3}{5}\lambda/\ln\lambda$ prime numbers between $\lambda$ and $2\lambda$, the          probability that a random prime number $p$ chosen in $(\lambda,2\lambda)$ is bad is at most
\[\frac{2.88\cdot 10^7(\ln\lambda)^{9/2}}{\frac{3}{5}\lambda/\ln\lambda}
    = 4.8\cdot 10^7\lambda^{-1}(\ln\lambda)^{11/2}.\]
The probability obviously tends to zero when $\lambda$ tends to infinity. For instance,  for $\lambda\ge 2^{55}$ the  probability is bounded by $2^{27}\lambda^{-1/2}$. Hence, to get a probability at most $\epsilon$, one can take        $\lambda \ge \frac{2^{54}}{\epsilon^2}$ (which is $> 2^{55}$ as long as $\epsilon \le\frac{1}{2}$).
\end{proof} 

From this effective result, we deduce a Monte Carlo algorithm that produces primes $p$, $q$ such that
$p\divides(q-1)$, as well as a $p$-PRU modulo $q$.  

\begin{proof}[Proof of \cref{lem:mcpap}]
This is basically Algorithm ``GetPrimeAP-5/6'' on page 35 of \citep{arn16}, 
slightly adapted, where the primality tests are made using AKS algorithm:
\begin{enumerate}[font=\bfseries\tiny, label=\theenumi]
\item sample $\le\frac{5}{6}\ln\frac{4}{\epsilon}\ln\lambda$ random odd integers $p\in(\lambda,2\lambda)$ until $p$ is prime, 
return \fail if none of them is prime \label{step:primep}
\item sample $\le 12\ln\frac{4}\epsilon\ln\lambda$ random even integers $a\in[1,\lambda^5]$ until $q = ap+1$ is prime,
return \fail if none of them is prime \label{step:primeq}
\item sample $\le\log_p\frac{4}{\epsilon}$ random elements $\zeta\in\GF q^\times$ until $\omega = \zeta^{(q-1)/p}\neq 1$,
return \fail if $\omega=1$ for each $\zeta$ \label{step:omega}
\item \textbf{return} $(p,q,\omega)$
\end{enumerate}
Since AKS has complexity $\polylog{\lambda}$ and $\log\frac{1}{\epsilon} = \oh{\log\lambda}$, the complexity of the whole algorithm is $\polylog(\lambda)$.

There are at least $\frac35\lambda/\ln\lambda$ primes in $(\lambda,2\lambda)$, and $\lambda/2$ odd integers. Therefore, the probability that a random odd integer is prime is at least $6/(5\ln\lambda)$. The probability that no prime is produced after $k$ tries is at most $(1-6/(5\ln\lambda))^k\le e^{-6k/(5\ln\lambda)}$. If $k = \frac{5}{6}\ln\frac{4}{\epsilon}\ln\lambda$, the probability is at most $\frac{\epsilon}{4}$. Hence Step~\ref{step:primep} succeeds with probability at least $1-\frac{\epsilon}{4}$. 

Since $\lambda\ge\frac{2^{54}}{(\epsilon/4)^2}$, if the algorithm succeeds in producing $p$, there are at least $\lambda^5/(24\ln\lambda)$ prime numbers $q\le\lambda^6$ of the form $ap+1$ with probability at least $1-\frac{\epsilon}{4}$. 

If $p$ satisfies this condition, there are at least $\lambda^5/(24\ln\lambda)$ values of $a$ such that $ap+1$, amongst the $\frac{1}{2}\lambda^5$ possible values. With the same proof as above, the probability that such an $a$ be found in $\le 12\ln\frac{4}{\epsilon}\ln\lambda$ tries is at least $1-\frac{\epsilon}{4}$. 

Finally, if $q$ has been found, the third step finds a suitable $\omega$ with probability at least
$1-\frac{\epsilon}{4}$ since there are at most $\frac{q-1}{p}$ values of $\zeta$ such that $\zeta^{(q-1)/p} = 1$.  

Therefore, the algorithm returns a triple $(p,q,\omega)$ satisfying the three properties with
probability at least $1-\epsilon$. 
\end{proof}

\begin{proof}[Proof of \cref{lem:mcpap2}]
Since $\lambda\ge\frac{2^{54}}{(\mu/2)^2}$, the prime $p$, if produced, satisfies 
that there are at least $\lambda^5/(24\ln\lambda)$ primes $q\le\lambda^6$ of the form $ap+1$ with
probability at least $1-\frac{\mu}{2}$. The number of those primes than can divide $K$
is at most $\log_\lambda K$ since all of them are $>\lambda$. Therefore, the probability
that one of them chosen at random divides $K$ is at most 
$24\log_pK\ln\lambda/\lambda^5\le\frac{\mu}{2}$.
\end{proof}

\section{Experiments}\label{sec:experiments}

\Cref{lem:mcpap,lem:mcpap2} are only valid for large values of $\lambda$. This is only an artefact due to the known explicit constants known for Bombieri-Vinogradov theorem. Actually, very similar results hold with smaller values. As an experimental justification of this, we perform the following computations.

A strong form of Dirichlet's theorem due to de la Vallée Poussin states that asymptotically, the proportion of primes in the arithmetic progression $\{2kp+1:1\le k\le p/2\}$ is $\oh{\frac{1}{\ln(p)}}$. For each prime number $p$ of bitsize between $10$ and $20$ (there are $81\,928$ of them), we estimate this proportion. For bitsizes $10$ to $14$, we actually compute the proportion exactly, testing the primality for each value $k$. For larger bitsizes, we estimate the proportion by sampling: We sample $N\ge 1000$ random elements in the set and test their primality.

The \textsc{SageMath} code used for the computations is given in Listing~\ref{listing}. 
\Cref{tab:data} provides, for each bitsize, the smallest and largest proportions found, as well as the average proportion. \Cref{fig:plot} plots the proportion for each prime.

\begin{lstlisting}[float,caption=Code for the proportion of primes in arithmetic progressions,label=listing]
def primes_in_arith_prog(p, bound, samples = 1000, random=True):
    """
    Estimate the number of primes <= bound in the arithmetic
    progression A = { 2pk+1 : k >= 1 }, by sampling samples
    random values k. The primality test is randomized if
    random = True.

    Note. If |A| <= bound, the number is computed exactly by 
    iterating over all elements of A.
    """
    count = 0
    size = bound // (2*p)
    if size > samples:
        for _ in range(samples):
            k = ZZ(randint(1, size))
            if (2*k*p+1).is_prime(proof=random):
                count += 1
        return RR(count/samples*size)

    for k in range(size):
        if ZZ(2*k*p+1).is_prime(proof=random):
            count += 1
    return count

def distrib_primes(bitsize, bound, samples = 1000, random=True):
    """
    Estimate the proportion of primes <= bound in each arithmetic
    progression A = {2kp+1 : k >= 1} for primes p of the given 
    bitsize.
    """
    L = []

    for p in prime_range(2**(bitsize-1),2**bitsize):
        n = primes_in_arith_prog(p,bound,samples,random)
        L.append(float(n/(bound//(2*p))))
    return L

def test(bmin, bmax):
    """
    This function generates the whole data set as a dictionary
    and prints the table that summarizes the result.
    """
    D = {}
    for b in range(bmin,bmax):
        L = distrib_primes(b,2**(2*b),samples=1000*(b-10+1))
        D[b] = L
    print("bitsize\tmin\taverage\tmax")
    for b in D:
        print(f"{b}\t{min(D[b]):.3f}\t{mean(D[b]):.3f}\t{max(D[b]):.3f}")
    return D

test(10, 21)
\end{lstlisting}

\begin{table}[p]
\caption{Proportion of primes in $\{2kp+1:1\le k\le p/2\}$.}
\label{tab:data}
\centering
\begin{tabular}{|c|c|c|c|c|}
\hline
Bitsize & Minimum & Average & Maximal & Theory\\\hline
10 & 13.36\% & 15.63\% & 18.49\% & 16.03\% \\
11 & 12.50\% & 14.12\% & 15.80\% & 14.43\% \\
12 & 11.00\% & 12.79\% & 14.33\% & 13.12\% \\
13 & 10.13\% & 11.79\% & 13.55\% & 12.02\% \\
14 & 9.50\% & 10.93\% & 12.54\% & 11.10\% \\
15 & 8.83\% & 10.14\% & 11.65\% & 10.30\% \\
16 & 8.24\% & 9.46\% & 10.74\% & 9.62\% \\
17 & 7.70\% & 8.89\% & 10.00\% & 9.02\% \\
18 & 7.19\% & 8.36\% & 9.41\% & 8.49\% \\
19 & 6.80\% & 7.91\% & 9.11\% & 8.01\% \\
20 & 6.53\% & 7.49\% & 8.52\% & 7.59\% \\
\hline
\end{tabular}
\end{table}

\begin{figure}[p]
\centering
\caption{Proportion of primes in $\{2kp+1:1\le k\le p/2\}$.}
\label{fig:plot}
\includegraphics[width=\textwidth]{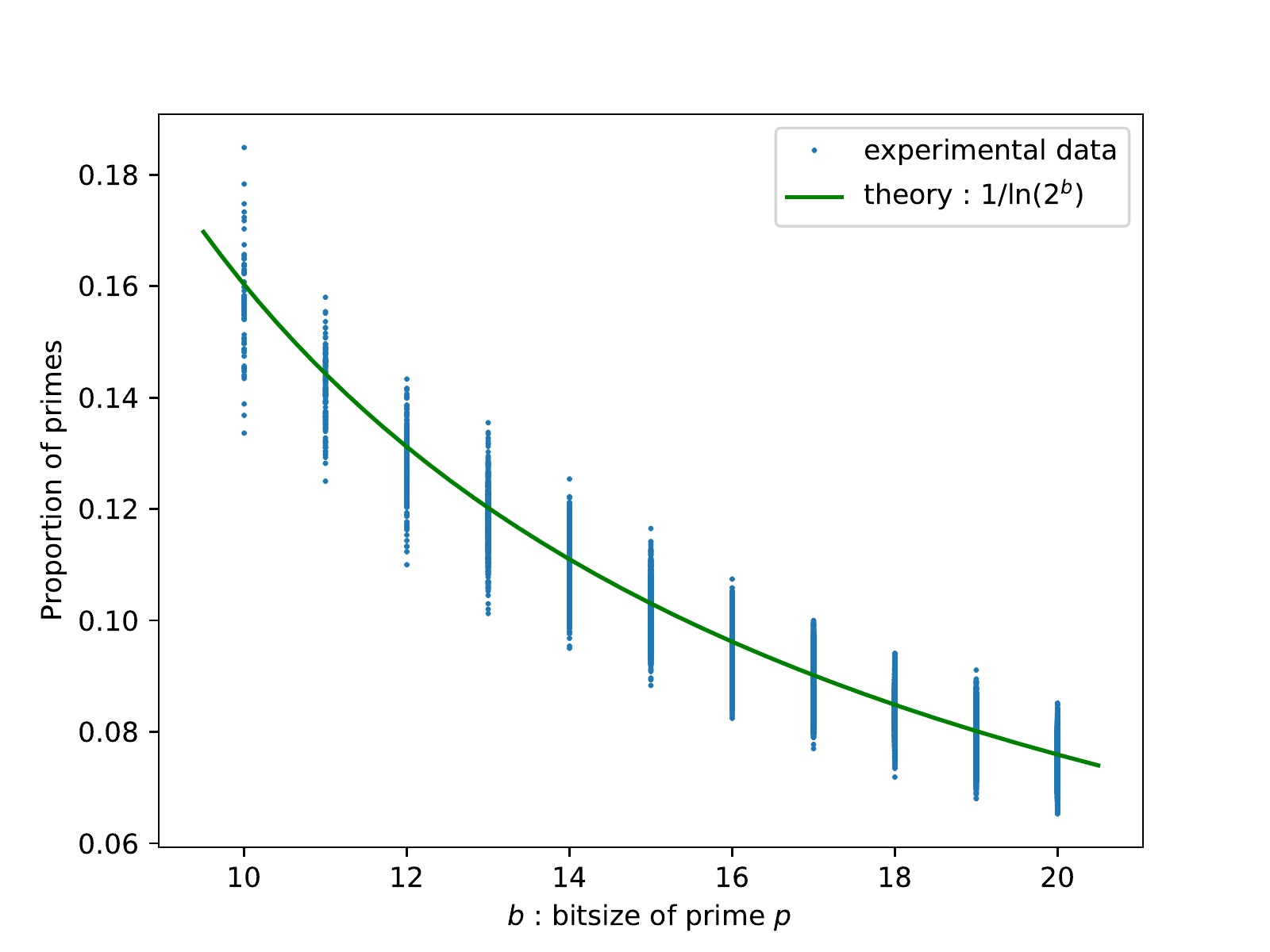}
\end{figure}

\bibliography{bib}

\end{document}